\documentclass[longbibliography, aps, pra, amsmath, amssymb, amsfonts, twocolumn,
superscriptaddress, footinbib]{revtex4-2}
\pdfoutput=1

\usepackage{graphicx}
\usepackage{subfigure}
\usepackage{dcolumn}
\usepackage{bm}
\usepackage{bbm}
\usepackage{amsthm}
\usepackage{thmtools}
\declaretheorem{theorem}
\usepackage{mathtools}
\usepackage{physics}
\usepackage{binarytree}
\usepackage{tikz,pgfplots}
\usepackage{verbatim}
\usepackage{pgfplots}
\usepackage{placeins}
\usepackage{algorithm2e}
\usetikzlibrary{decorations.pathreplacing}
\usetikzlibrary{arrows.meta}
\usetikzlibrary{graphs}
\usepackage{xcolor}
\usepackage{rotating}
\usepackage{yquant}
\usepackage{tabularray}
\usepackage[colorlinks, linkcolor=red, anchorcolor=blue, citecolor=green]{hyperref}
\RestyleAlgo{ruled}
\pgfplotsset{compat=1.16}
\begin{document}

\title[Quantum mutual information redistribution by Number Partitioning algorithm] {Quantum mutual information redistribution by Number Partitioning algorithm}

\author{Muchun Yang}

\affiliation{Institute of Physics, Beijing National Laboratory for Condensed
  Matter Physics,\\Chinese Academy of Sciences, Beijing 100190, China}

\affiliation{School of Physical Sciences, University of Chinese Academy of
  Sciences, Beijing 100049, China}

\author{Cheng-Qian Xu}

\affiliation{Institute of Physics, Beijing National Laboratory for Condensed
  Matter Physics,\\Chinese Academy of Sciences, Beijing 100190, China}

\affiliation{School of Physical Sciences, University of Chinese Academy of
  Sciences, Beijing 100049, China}

\author{D. L. Zhou} \email[]{zhoudl72@iphy.ac.cn}

\affiliation{Institute of Physics, Beijing National Laboratory for Condensed
  Matter Physics,\\Chinese Academy of Sciences, Beijing 100190, China}

\affiliation{School of Physical Sciences, University of Chinese Academy of
  Sciences, Beijing 100049, China}

\date{\today}

\begin{abstract}
  Quantum information distribution in a tripartite state plays a fundamental role in quantum information processes. Here we investigate how a bipartite unitary transformation $U_{AB}$ redistributes the quantum mutual information with the third party $C$ in a tripartite pure state $|\psi\rangle_{ABC}$ in a $d_A\times d_B\times d_C$ dimensional Hilbert space. In particular, we focus on finding out the optimal unitary transformation $U_{AB}^{\ast}$ that maximizes the quantum mutual entropy between party $A$ and party $C$, $I(A:C)=S(\rho_A)-S(\rho_B)+S(\rho_C)$. We show that the mutual entropy $I(A:C)$ is upper bounded by $2S(\rho_C)$ derived from the Araki-Lieb inequality. This upper bound can be realized via an optimal unitary transformation for any pure state with the rank $r_{C}$ of $\rho_C$ satisfying $r_C\le d_A$. For a generic pure state with $r_C> d_A$, the upper bound can not be realized by any bipartite unitary transformation. To maximize the mutual entropy in the latter case, we propose a fast numerical algorithm to produce an approximate optimal unitary transformation, where our optimization is transformed into a modified number partition problem. The validness of our algorithm is confirmed by its comparison with the results from the Adam algorithm for parameterized unitary transformations. Our approximate algorithm thus provides a practical protocol to implement redistribution of quantum mutual information for a tripartite quantum state with high dimensions.
\end{abstract}

\maketitle

\section{Introduction}

Quantum information distribution in a tripartite state is fundamental in quantum information processes, where the strong subadditivity inequality of Von-Neumann entropy gives important constrains on these information distribution. Here we investigate how a bipartite unitary transformation changes the quantum information distribution in a tripartite pure state.

To characterize the quantum information distribution, we use quantum mutual information as a basic measure. Mutual information in Shannon's theory is a fundamental quantity of information transmitting capacity~\cite{https://doi.org/10.1002/j.1538-7305.1948.tb01338.x}. Quantum mutual information is a measure of correlation with information transmission task in quantum state~\cite{PhysRevLett.83.3081,quan_info_meets_quan_matter}. In quantum on-time pad, quantum mutual information is the maximum information that can securely send~\cite{PhysRevA.74.042305,PhysRevA.104.022610}. Quantum mutual information also quantifies the minimal amount of noise needed to erase the correlation in a bipartite state~\cite{PhysRevA.72.032317,PhysRevLett.121.040504}.

In a tripartite pure state, a bipartite unitary transformation can not change the mutual entropy between the two parties and the third party, but it changes the distribution of quantum mutual entropy among them. Thus our aim is to maximize the mutual entropy of the first party and the third one, which is shown to be equivalent to the maximization of the entropy difference of the first party relative to the second party.

The basic application of our mutual information redistribution is as follows. Assume Alice, Bob, and Charlie share an entangled state $|\psi\rangle_{ABC}$, which is used as a quantum resource for performing securely transfer classical information between Alice (and/or Bob) and Charlie. Alice and Bob can redistribute their capacities by performing an optimized bipartite unitary transformation, although their total capacity keeps invariant.

To solve the above maximization, we develop an approximate numerical algorithm, which can be transformed into a modified number partition problem. Number partitioning problem~(NPP) is to partition a group numbers into a fixed number of subsets, such that the sums of each subset are as similar as possible. Finding the exact solution is difficult, which is an NP-hard problem~\cite{Korf1995FromAT,article1}. There are lots of approximate algorithms to give approximate solutions ~\cite{https://doi.org/10.1111/0824-7935.00069,doi:10.1137/0117039,Xiao2017/04,10.1145/800200.806205,Chen1993ANO,doi:10.1137/0603019,CSIRIK1992281,WU2005407}. Recently, physicists also proposed quantum algorithm for NPP~\cite{Graß2016,PRXQuantum.2.020319,sinitsyn2023topologically}. Our algorithm is based on the maximization of the entropy of the first party before the minimization of the second party, which are related with the concavity~\cite{nielsen_chuang_2010} and the majorization~\cite{LI2013384,an2021learning} properties of Von-Neumann entropy respectively. Here the maximization of the entropy of the first party is mapped to a modified number partition problem after a disentanglement unitary transformation. The validness of our algorithm is confirmed by its comparison with the results from the Adam algorithm for parameterized unitary transformations. Our approximate algorithm thus provides a practical protocol to implement redistribution of quantum mutual information for a tripartite quantum state.

\section{The maximization problem of quantum mutual information}
\label{sec:our-main-problem}

For a bipartite quantum state $\rho_{AB}$, the quantum mutual information between $A$ and $B$ is defined as
\begin{align}
    I_{\rho}(A:B)\equiv S(\rho_A)+S(\rho_B)-S(\rho_{AB}),
\end{align}
where the state function $S$ is the von-Neumann entropy: for a quantum state $\sigma$, $S(\sigma)\equiv -\Tr[\sigma\log \sigma]$, and the logarithm is taken on base $2$. The states $\rho_A$ and $\rho_B$ are reduced density matrices of the state $\rho_{AB}$.

For a pure tripartite state $|\Psi\rangle_{ABC}$, a direct calculation leads to the following equality on quantum mutual information:
\begin{align}\label{IABC}
    I_{|\psi\rangle}(AB:C)=I_{|\psi\rangle}(A:C)+I_{|\psi\rangle}(B:C)=2S(\rho_C),
  \end{align}
  where
\begin{align}
    I_{|\psi\rangle}(A:C) &=S(\rho_C)+S(\rho_A)-S(\rho_B),\label{eq:5}\\
    I_{|\psi\rangle}(B:C) &=S(\rho_C)+S(\rho_B)-S(\rho_A).\label{eq:7}
\end{align}
The states $\rho_{A}$, $\rho_B$ and $\rho_C$ are reduced states of $|\psi\rangle_{ABC}$. Eq.~\eqref{IABC} implies that for a pure tripartite state  the quantum mutual information between $C$ and $\{A,B\}$ is completely distributed into the quantum mutual information between $C$ and $A$ and that between $C$ and $B$.

Our main task can be formulated as follows. Assume that Alice, Bob and Charlie share a pure tripartite state $|\psi\rangle_{ABC}$, which is the main resource assisting the communications between Charlie and Alice (or Bob). By performing a unitary transformation $U_{AB}$ between Alice and Bob, the quantum mutual information between Charlie and Alice (or Bob) can be adjustable. Our aim is to maximize the quantum mutual information
\begin{equation}
  \label{eq:1}
  \max_{U_{AB}} I_{U_{AB}|\psi\rangle_{ABC}}(A:C) \equiv I_M.
\end{equation}
Because $U_{AB}$ is a unitary transformation on $AB$ acting on $AB$, the mutual information between $C$ and $AB$ is invariant, i.e.,
\begin{equation}
  \label{eq:2}
  I_{U_{AB}|\psi\rangle_{ABC}}(AB:C) = I_{|\psi\rangle_{ABC}}(AB:C) = 2 S(\rho_{C}).
\end{equation}
This implies that the unitary transformation that maximizes the mutual information between $A$ and $C$ must minimize the mutual information between $B$ and $C$:
\begin{equation}
  \label{eq:3}
  \min_{U_{AB}} I_{U_{AB}|\psi\rangle_{ABC}}(B:C) = 2 S(\rho_C) - I_M.
\end{equation}
In particular, when the dimension of the Hilbert space of party $A$ equals to that of party $B$, the range of the mutual information between $A$ and $C$ under any unitary transformation $U_{AB}$ is given by
\begin{equation}
  \label{eq:4}
  2 S(\rho_C) - I_M \le I_{U_{AB}|\psi\rangle_{ABC}}(A:C) \le I_M.
\end{equation}
Following Eq.~\eqref{eq:5}, the maximization in Eq.~\eqref{eq:1} can be simplified as
\begin{equation}
  \label{eq:8}
  \max_{U_{AB}} [S(\rho_A^U)-S(\rho_B^U)] = I_M - S(\rho_C),
\end{equation}
where
\begin{align}
  \label{eq:9}
  \rho^0_{AB} & = \Tr_C (|\psi\rangle_{ABC}\, {}_{{ABC}}\langle\psi|), \\
  \rho_A^U & = \Tr_{B} (U_{AB} \rho^0_{AB} U_{AB}^{\dagger}), \\
  \rho_B^U & = \Tr_{A} (U_{AB} \rho^0_{AB} U_{AB}^{\dagger}).
\end{align}
We observe that $\rho^0_{AB}$ determines the result of the maximization. The equivalent maximization in Eq.~\eqref{eq:1} and that in Eq.~\eqref{eq:8} are demonstrated in Fig.~\ref{figure0}.

\begin{figure}[htbp]
    \centering
\begin{tikzpicture}
  \begin{yquant*}
    qubit {} q[3];
    init {$|\psi\rangle_{ABC}$} (q[0],q[1],q[2]);
    box {$U_{AB}$} (q[0],q[1]);
    output {$\max_{U_{AB}} I_{U_{AB}|\psi\rangle_{ABC}}(A:C)$} (q[0],q[1]);
  \end{yquant*}
\end{tikzpicture} \\
    \centering
\begin{tikzpicture}
 \draw[<->,blue,thick] (2,0) -- (2,-1);
\end{tikzpicture} \\
    \centering
  \begin{tikzpicture}
  \begin{yquant*}
    qubit {} q[2];
    init {$\rho_{AB}$} (q[0],q[1]);
    box {$U_{AB}$} (q[0],q[1]);
    output {$\max [S(\rho_A^U)-S(\rho_B^U)]$} (q[0],q[1]);
  \end{yquant*}
  \end{tikzpicture}
    \caption{Equivalent pictures of redistribution of quantum mutual information.}
    \label{figure0}
\end{figure}
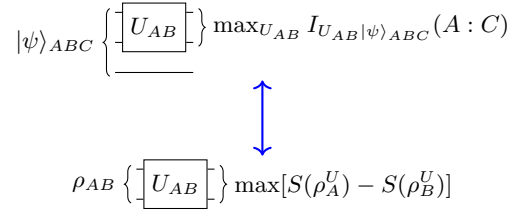

\section{Upper bound of entropy difference and Araki-Lieb inequality}

Eq.~\eqref{eq:8} makes us reminiscent of the Araki-Lieb inequality~\cite{Araki1970,Carlen2012}
\begin{align}
    |S(\rho_A)-S(\rho_B)|\leq S(\rho_{AB}).  \label{eq:6}
\end{align}
In our case, $S(\rho_{AB})=S(\rho_C)$, which is invariant under any unitary transformation $U_{AB}$. Thus $S(\rho_A^U)-S(\rho_B^U)\le S(\rho_C)$, which implies that the entropy difference $\Delta S$ is upper bounded by $S(\rho_C)$. Let us investigate under which condition the upper bound $S(\rho_C)$ for $\Delta S$ can be arrived at. Our main result is summarized in the following theorem.

\begin{theorem}
\label{thm2}
For a tripartite pure state $|\psi_{ABC}\rangle$ with Hilbert space dimension $d_A\times d_B \times d_C$, there exists an optimal unitary transformation $U_{AB}^{\ast}$ such that  $S(\rho_A^{U^{\ast}})-S(\rho_B^{U^{\ast}})=S(\rho_C)$ if and only if the rank of $\rho^{C}$: $r_{C}\le d_A$.
\end{theorem}

\begin{proof}
Let us first prove the ``only if'' part. If the upper bound can be arrived at for a state $|\psi\rangle_{ABC}$, then there exists a unitary transformation $U^{\ast}_{AB}$, such that
\begin{equation}
  \label{eq:47}
 S(\rho_A^{U^{\ast}})-S(\rho_B^{U^{\ast}})=S(\rho_C).
\end{equation}
 Let us introduce the notations:
\begin{align}
  \label{eq:39}
  |\psi^{\ast}\rangle_{ABC} & = U^{\ast}_{AB} |\psi\rangle_{ABC}, \\
  \rho^{\ast}_{ABC} & = |\psi^{\ast}\rangle_{ABC} {}_{ABC}\langle\psi^{\ast}|.
\end{align}
Then the mutual information
\begin{align}
  \label{eq:48}
  I_{\rho^{\ast}}(A:C)& =2S(\rho_C), \\
  I_{\rho^{\ast}}(B:C)& =0. \label{eq:49}
\end{align}
The latter equality implies that
\begin{equation}
  \label{eq:40}
  \rho^{\ast}_{BC} = \rho^{\ast}_B \otimes \rho_C = \sum_{m=1}^{r_B} \sum_{n=1}^{r_C} p_{m} q_n |m n\rangle \langle m n|,
\end{equation}
where $\{|m\rangle\}$ and $\{|n\rangle\}$ are orthogonal normal bases, $p_m, q_n>0$, $\sum_m p_m=\sum_n q_n=1$, $r_B$ and $r_C$ are the ranks of $\rho^{\ast}_B$ and $\rho_C$ respectively. Thus
\begin{equation}
  \label{eq:41}
  |\psi^{\ast}\rangle_{ABC} = \sum_{m=1}^{r_B} \sum_{n=1}^{r_C} \sqrt{p_{m}q_n} |\phi_{mn}\rangle \otimes |m\rangle \otimes |n\rangle,
\end{equation}
where $\{|\phi_{mn}\rangle\}$ is the orthogonal normal basis of the Hilbert space of party $A$. Eq.~\eqref{eq:41} implies that
\begin{align}
  \label{eq:42}
  r_B & \le d_B, \\
  r_B r_C & \le d_A.
\end{align}
Because $r_B\ge1$, we have
\begin{equation}
  \label{eq:45}
  r_C  \le d_A.
\end{equation}

Then we prove the ``if'' part. Our tripartite pure state can be written as
\begin{equation}
  \label{eq:50}
  |\psi_{ABC}\rangle = \sum_{n=1}^{r_C} \sqrt{q_n} |\chi_n\rangle \otimes |n\rangle,
\end{equation}
where $q_n>0$, $\sum_n q_n=1$, $\{|n\rangle\}$ and $\{|\chi_n\rangle\}$ are orthonormal base of Hilbert space of $C$ and $AB$ respectively. Then there exist a unitary transformation such that
\begin{equation}
  \label{eq:51}
  U_{AB}^{\ast} |\chi_n\rangle = |\phi_n\rangle \otimes |1\rangle,
\end{equation}
where $\{|\phi_n\rangle\}$ are orthonormal basis of Hilbert space of $A$, and $|1\rangle$ is a normalized state of $B$. Then
\begin{equation}
  \label{eq:52}
  \rho_{AB}^{\ast} = \sum_{n=1}^{r_C} q_n |\phi_n\rangle \langle \phi_n| \otimes |1\rangle \langle1|,
\end{equation}
which satisfies $S(\rho_A^{\ast})-S(\rho_B^{\ast})=S(\rho_C)$. This completes the  proof of our theorem.
\end{proof}

Because $r_C\le d_C$, we get a direct corollary of the above theorem: if the dimension of Hilbert spaces $A$ and $C$ satisfies $d_A>d_C$, there exists an optimal unitary transformation $U^{\ast}$ such that $S(\rho_A^{\ast})-S(\rho_B^{\ast})=S(\rho_C)$.

For simplicity, we focus on the case with $d_A=d_B=d$ in the following. For a general stat $\rho^{AB}$ in a $d\times d$ Hilbert space, the rank of $\rho^{AB}$ lies the range $[1,d^2]$. Only the states with lower rank in the range $[1,d]$ can arrive at the upper bound $S(\rho_{AB})$ from the Araki-Lieb inequality. For the states with the rank in the range $[d+1,d^2]$, the upper bound $S(\rho_{AB})$ can not arrive at.

When the upper bound of entropy difference can be reached, Eqs.~\eqref{eq:48}\eqref{eq:49} implies that the optimal unitary transformation completely transforms the correlation between $AB$ and $C$ to that between $A$ and $C$, without any correlation left between $B$ and $C$. The upper bound can not be reached by a unitary transformation originates from the limitation of the dimension of the Hilbert space.

\section{Solving mutual information maximization by number partition algorithm}
\label{sec:solv-mutu-inform}

The eigen-decomposition of $\rho_{AB}^{0}$ is
\begin{equation}
    \rho_{AB}^{0}=\sum_{m=0}^{d-1} \sum_{n=0}^{d-1} p_{m n}|\psi_{m n}\rangle\langle\psi_{m n}|.
\end{equation}
We assume these eigenvalues are in the decreasing order, i.e., $p_{m n} \ge p_{m^{\prime} n^{\prime}}$ if $md+n\le m^{\prime}d+n^{\prime}$. First we apply the disentanglement unitary transformation
\begin{equation}
  \label{eq:10}
  D |\psi_{m n}\rangle = |m n\rangle,
\end{equation}
which makes $\rho^0_{AB}$ become a separable state
\begin{equation}
  \label{eq:11}
  \rho^D_{AB} \equiv D \rho^0_{A B} D^{\dagger} = \sum_{m=0}^{d-1} \sum_{n=0}^{d-1} p_{m n} |m n\rangle \langle m n|.
\end{equation}
Then we define a unitary transformation related to a permutation operation $s$
\begin{equation}
  \label{eq:12}
  U_s |m n\rangle = |s(m n)\rangle,
\end{equation}
where $s$ is an element in the permutation group $S_{d^2}$. Thus
\begin{align}
  \label{eq:10}
  \rho^s_{AB} & \equiv U_s \rho^D_{AB} U_s^{\dagger}   = \sum_{m=0}^{d-1} \sum_{n=0}^{d-1} p_{m n} |s(m n)\rangle \langle s(m n)| \nonumber\\
  & = \sum_{m=0}^{d-1} \sum_{n=0}^{d-1} p_{s^{-1}(m n)} |m n\rangle \langle m n|,
\end{align}
where $s^{-1}$ is the inverse of $s$. The reduced states of $\rho^s_{AB}$ are
\begin{align}
  \label{eq:13}
  \rho^s_A & =  \sum_{m=0}^{d-1} p^s_{A m} |m \rangle \langle m |, \\
  \rho^s_B & = \sum_{n=0}^{d-1}  p^s_{B n} | n\rangle \langle  n|.
\end{align}
with
\begin{align}
  \label{eq:19}
  p^s_{A m} & = \sum_{n=0}^{d-1} p_{s^{-1}(m n)}, \\
  p^s_{B n} & = \sum_{m=0}^{d-1} p_{s^{-1}(m n)}.
\end{align}
Hence
\begin{align}
  \label{eq:14}
  S(\rho^s_A) - S(\rho^s_B) & = - \sum_{m=0}^{d-1} p^s_{A m} \log p^s_{A m} \nonumber\\
  & + \sum_{n=0}^{d-1} p^s_{B n} \log p^s_{B n}.
\end{align}
Thus the maximization over the permutation operation is given by
\begin{equation}
  \label{eq:16}
  \max_{s\in S_{d^2}} [S(\rho^s_A) - S(\rho^s_B)].
\end{equation}
In the following, we aim to show that Eq.~\eqref{eq:16} is an excellent substituent of Eq.~\eqref{eq:8} in most cases for our optimization problem.

Before detailed numerical optimization, we explore the symmetry in our problem. Let $r\in S_d$ and $t\in S_d$, and $r\otimes t\in S_d\otimes S_d$, and $S_d\otimes S_d$ is a subgroup of $S_{d^2}$:
\begin{equation}
  \label{eq:15}
  r \otimes t (m n) = (r(m) t(n)).
\end{equation}
Thus we can prove that
\begin{align}
  \label{eq:17}
  S(\rho_A^{r\otimes t \cdot s}) & = S(\rho_A^s), \\
  S(\rho_B^{r\otimes t \cdot s}) & = S(\rho_B^s).\label{eq:20}
\end{align}
The proof of Eq.~\eqref{eq:17} is given as follows:
\begin{align}
  \label{eq:18}
  & S(\rho^{r\otimes t \cdot s}_A)   = - \sum_{m=0}^{d-1} \sum_{n=0}^{d-1} p_{(r\otimes t \cdot s)^{-1}(m n)} \log \sum_{n=0}^{d-1} p_{(r\otimes t \cdot s)^{-1}(m n)} \nonumber\\
  & = - \sum_{m=0}^{d-1} \sum_{n=0}^{d-1} p_{s^{-1} \cdot (r\otimes t)^{-1}(m n)} \log \sum_{n=0}^{d-1} p_{s^{-1} \cdot (r\otimes t)^{-1}(m n)} \nonumber\\
  & = - \sum_{m=0}^{d-1} \sum_{n=0}^{d-1} p_{s^{-1}(r^{-1}(m) t^{-1}(n))} \log \sum_{n=0}^{d-1} p_{s^{-1}(r^{-1}(m) t^{-1}(n))} \nonumber\\
& = - \sum_{m^{\prime}=0}^{d-1} \sum_{n^{\prime}=0}^{d-1} p_{s^{-1}(m^{\prime} n^{\prime})} \log \sum_{n^{\prime}=0}^{d-1} p_{s^{-1}(m^{\prime} n^{\prime})} \nonumber\\
& = S(\rho_A^{s}).
\end{align}
Similarly, we can prove Eq.~\eqref{eq:20}. Thus for any $r, t\in S_d$ and $s\in S_{d^2}$,
\begin{equation}
  \label{eq:21}
  S(\rho_A^{r\otimes t \cdot s}) - S(\rho_B^{r\otimes t \cdot s})  = S(\rho_A^s) - S(\rho_B^{s}),
\end{equation}
which implies that every element $s$ in one right coset of the subgroup $S_d\otimes S_d$ will give  the same value of $S(\rho_A^s)-S(\rho_B^s)$. In other words, the maximization in Eq.~\eqref{eq:16} is taken over the set of all the right cosets of the subgroup $S_d\otimes S_d$, any one element in each coset:
\begin{equation}
  \label{eq:30}
  \max_{s\in S_{d^2}/S_d\otimes S_d} [S(\rho^s_A) - S(\rho^s_B)].
\end{equation}

In particular, we realize that the maximization problem in Eq.~\eqref{eq:16} or in Eq.~\eqref{eq:30} is a type of NPP: $d^2$ numbers $\{p_{m n}\}$ are partitioned into a $d\times d$ lattice, every site with one element. Every partition corresponds to a permutation element. Our maximization is taken over all the ways of partitions. The case of $d=3$ is demonstrated in Table~\ref{tab:1}. To maximize $S(\rho_A^s)-S(\rho_B^s)$, we need to choose a permutation $s$ such that the numbers in $\{p^s_{A m}\}$ are as similar as possible and the numbers in $\{p_{B n}^s\}$ are as different as possible.

Let us consider the number of the permutations in the set $S_{d^2}$ or in the set of $S_{d^2}/S_d\otimes S_d$, which is $(d^2)!$ or $(d^2)!/(d!d!)$. With the increasing of $d$, there numbers become extremely large, e.g., $(4^2)!\simeq2.1\times10^{13}$ and $(4^2)!/(4!4!)\simeq3.6\times10^{10}$, which prevents the numerical optimization directly by the exhaustive attack method.

\begin{table}[htbp]
  \centering
  \begin{tblr}{|c|ccc|}
    \hline
    & $p_{B0}^s$ &  $p_{B1}^s$ & $p_{B2}^s$ \\
    \hline
    $p_{A0}^s$ & $p_{s^{-1}(00)}$ &$p_{s^{-1}(01)}$ &$p_{s^{-1}(02)}$ \\
    $p_{A1}^s$ & $p_{s^{-1}(10)}$ &$p_{s^{-1}(11)}$ &$p_{s^{-1}(12)}$ \\
    $p_{A2}^s$ & $p_{s^{-1}(20)}$ &$p_{s^{-1}(21)}$ &$p_{s^{-1}(22)}$ \\
    \hline
  \end{tblr}
  \caption{Number partition of $\{p_{mn}\}$ with column sums and row sums for $d=3$.}\label{tab:1}
\end{table}

\subsection{$d=2$ case}
\label{sec:d=2-case}

Let us start with the case of $d=2$, where both A and B are one qubit. Let us take a computational basis of $\mathcal{H}_A$ as $ \{|m\rangle,0\le a\le 1\} $ (the eigenvectors of $Z_A$ with eigenvalues $(-1)^{m}$), and a basis of $\mathcal{H}_B$ as $ \{|n\rangle,0\le n\le 1\} $ (the eigenvectors of $Z_B$ with eigenvalues $(-1)^n$). Then for any permutation $s$ we construct a unitary transformation
\begin{align}
  U_s|m n\rangle &= |s(m n)\rangle.
\end{align}
The number of all the unitary transformations related with permutations equals to the order of the permutation group $S_4$, i.e., $4!=24$. We can show that
\begin{align}
  \label{eq:25}
  U_s Z_A U_s^{\dagger} & = (-1)^{c^A_s} Z_A^{a^A_s} Z_B^{b^A_s}, \\
  U_s Z_B U_s^{\dagger} & = (-1)^{c^B_s} Z_A^{a^B_s} Z_B^{b^B_s}. \label{eq:26}
\end{align}

Note that the unitary transformations related with the subgroup $S_2\otimes S_2$ is
\begin{equation}
  \label{eq:22}
  X_A^a \otimes X_B^b, ~~~a,b\in\{0,1\},
\end{equation}
where $X$ is the $x$ component of the Pauli operator defined by
\begin{equation}
  \label{eq:23}
  X |m\rangle = |1-m\rangle,
\end{equation}
or equivalently defined by
\begin{equation}
  \label{eq:24}
  X Z X^{\dagger} = - Z.
\end{equation}
Then Eqs.~\eqref{eq:25}\eqref{eq:26} become
\begin{align}
  \label{eq:27}
  X_A^a X_B^b Z_A X_B^b X_A^a & = (-1)^a Z_A, \\
  X_A^a X_B^b Z_B X_B^b X_A^a & = (-1)^b Z_B.
\end{align}
Thus the representative element in the unitary transformations corresponding to the right cosets of the subgroup $S_2\otimes S_2$ are given by
\begin{align}
  \label{eq:28}
  U_s Z_A U_s^{\dagger} & = Z_A^{a^A_s} Z_B^{b^A_s}, \\
  U_s Z_B U_s^{\dagger} & = Z_A^{a^B_s} Z_B^{b^B_s}.
\end{align}
which means the unitary transformations of $\{Z_A, Z_B\}$ can take the values of any two ordered elements in $\{Z_A, Z_B, Z_A Z_B\}$. This implies there are $P_3^2=6$ representative unitary transformations in the cosets. If we denote $s(m n)=(m^{\prime} n^{\prime}$), then these unitary transformations can be obtained by solving the following equations:
\begin{align}
  \label{eq:29}
  s^{-1}
  \begin{pmatrix}
    m^{\prime} \\ n^{\prime}
  \end{pmatrix}
  =
  \begin{pmatrix}
    m \\ n
  \end{pmatrix}
  \mod 2,
\end{align}
with
\begin{equation}
  \label{eq:31}
  s^{-1} =
  \begin{pmatrix}
    a_s^A & b_s^A \\
    a_s^B & b_s^B
  \end{pmatrix}.
\end{equation}
The above equations has a unique solution if and only if
\begin{equation}
  \label{eq:30}
  \det s^{-1} = a^A_s b^B_s - a^B_s b^A_s \neq 0.
\end{equation}
Then all the permutations in $S_4/S_2 \otimes S_2$  are given by
\begin{align}
  \label{eq:32}
  {s_1}^{-1} & =
  \begin{pmatrix}
    1 & 0 \\
    0 & 1
  \end{pmatrix}, \\
  {s_2}^{-1} & =
  \begin{pmatrix}
    0 & 1 \\
    1 & 0
  \end{pmatrix}, \\
  {s_3}^{-1} & =
  \begin{pmatrix}
    1 & 0 \\
    1 & 1
  \end{pmatrix}, \\
  {s_4}^{-1} & =
  \begin{pmatrix}
    1 & 1 \\
    0 & 1
  \end{pmatrix}, \\
  {s_5}^{-1} & =
  \begin{pmatrix}
    0 & 1 \\
    1 & 1
  \end{pmatrix}, \\
  {s_6}^{-1} & =
  \begin{pmatrix}
    1 & 1 \\
    1 & 0
  \end{pmatrix}.
\end{align}
Their inverse permutations are given by
\begin{align}
  \label{eq:33}
  s_i & = {s_i}^{-1}, \; 1\le i\le 4, \\
  s_5 & = {s_6}^{-1}, \\
  s_6 & = s_5^{-1}.
\end{align}
A direct calculation gives
\begin{align}\label{s5}
     \underset{{s\in S_{4}/S_2\otimes S_{2}}}{\operatorname{argmax}} [S(\rho^s_A) - S(\rho^s_B)] = s_5,
\end{align}
where
\begin{equation}\label{s_2bit}
s_5^{-1}
\begin{pmatrix}
  m \\ n
\end{pmatrix}
=
\begin{pmatrix}
  n \\ m + n \mod 2
\end{pmatrix}.
\end{equation}
Then the state
\begin{align}
  \rho^{s_{5}}_{AB} &    = \sum_{m=0}^{1} \sum_{n=0}^{1} p_{s_{5}^{-1}(m n)} |m n\rangle \langle m n| \nonumber\\
  & = p_{00} |00\rangle \langle 00| + p_{11} |01\rangle \langle01| + p_{01} |10\rangle \langle10| + p_{10} |11\rangle \langle 11|.
\end{align}
The reduced states of $\rho^{s_5}_{AB}$ are

\begin{align}
    \rho^{s_{5}}_A &= (p_{00}+p_{11}) |0\rangle_{A} \langle0| + (p_{01}+p_{10}) |1\rangle_A \langle 1|,   \\
    \rho^{s_{5}}_B &= (p_{00}+p_{01}) |0\rangle_{B} \langle0| + (p_{11}+p_{10}) |1\rangle_B \langle 1|.
\end{align}
Hence the maximum value

\begin{align}
    S(\rho^{s_{5}}_A) - S(\rho^{s_{5}}_B) = & -(p_{00}+p_{11})\log (p_{00}+p_{11})  \nonumber \\
    &-(p_{01}+p_{10})\log (p_{01}+p_{10}) \nonumber \\ 
    & +(p_{00}+p_{01})\log (p_{00}+p_{01}) \nonumber \\ 
    &+(p_{10}+p_{11})\log (p_{10}+p_{11}).
\end{align}

This permutation is demonstrated in Table~\ref{tab:2}.

\begin{table}[htbp]
  \centering
  \begin{tblr}{|c|cc|}
    \hline
    & $p_{B0}^{s_5}$ &  $p_{B1}^{s_5}$  \\
    \hline
    $p_{A0}^{s_5}$ & $p_{00}$ &$p_{11}$\\
    $p_{A1}^{s_5}$ & $p_{01}$ &$p_{10}$\\
    \hline
  \end{tblr}
  \caption{Optimal number partition of $\{p_{mn}\}$ with column sums and row sums for $d=2$, which corresponds to Eq.~\eqref{s5}}.\label{tab:2}
\end{table}

\begin{theorem}
\label{thm1}
For a $2$-qubit state $\rho_{AB}$, the unitary transformation $U^{\ast}=U_{s_{5}}D$ makes $S(\rho_A^U)-S(\rho_B^{U})$ take a local maximum.
\end{theorem}

\begin{proof}
  Let $U=W U^{\ast}$. Then
\begin{align}
  \Delta S(U) & \equiv S(\rho_A^U) - S(\rho_B^U) \notag \\
  & = S(\Tr_B [W\rho_{AB}^{s_5}W^{\dagger}])- S(\Tr_A [W\rho_{AB}^{s_5}W^{\dagger}]),
\end{align}
where the unitary transformation $W$ can be parameterized as
\begin{align}
  W &= W_{00}(h_{00}) \prod_{j=1}^3 W_{j0}(h_{j0}) \notag \\
  & \phantom{=} \times \prod_{k=1}^3 W_{0j}(h_{0j}) \prod_{m,n=1}^{3} W_{mn}(h_{mn})
\end{align}
with
\begin{equation}
  \label{eq:36}
  W_{m n}(h_{m n}) =\exp(ih_{m n}\sigma_m\otimes \sigma_n).
\end{equation}
Let
\begin{equation}
  \label{eq:34}
  W^{\prime} = \prod_{m,n=1}^{\prime 3} W_{mn}(h_{mn})
\end{equation}
and $U^{\prime}=W^{\prime}U^{\ast}$. Then we can prove
\begin{equation}
  \label{eq:35}
  \Delta S(U) = \Delta S(U^{\prime}) = \Delta S(\{h_{m n}\}^{\prime}),
\end{equation}
where ${}^{\prime}$ implies that it does not contain the term with $m=n=3$.

Then we need to show that
\begin{align}
  \label{eq:77}
  \eval{\pdv{\Delta S}{h_{m n}}}_{\{h_{mn}=0\}} & =0, \\
  \eval{H(\Delta S)}_{\{h_{mn}=0\}} & \le 0,\label{eq:78}
\end{align}
with the Hessian matrix
\begin{equation}
  \label{eq:37}
  H(\Delta S)_{m n; m^{\prime} n^{\prime}} = \pdv{\Delta S}{h_{mn}}{h_{m^{\prime} n^{\prime}}}.
\end{equation}
Eqs.~\eqref{eq:77}\eqref{eq:78} shows that  $U^{\ast}$ is a local maximum of $\Delta S(U)$, whose proofs are given in Appendix~\ref{app:one}.
\end{proof}

There remains a question that, what is the global maximal value of entropy difference $\Delta S(U)$? Fortunately, we have not found any other unitary transformation that get larger $\Delta S(U)$.  We can use the Adam optimizing algorithm~\cite{Kingma2014AdamAM}, which is a Gradient Descending (GD) method, to optimize function $\Delta S(U)$.
Fig.~\ref{2bit_and_3bit}~(a)(b) shows that the numerical results using Adam optimization are actually the same as that using $U^{\ast}$, where the relative error is in the order of $10^{-8}$ numerically. Here the relative error is defined by $(\Delta S_{\text{Adam}}-\Delta S_{\text{permutation}})/\Delta S_{\text{Adam}}$.
 But we failed to prove that $U^{\ast}$ is the global optimal unitary transformation.

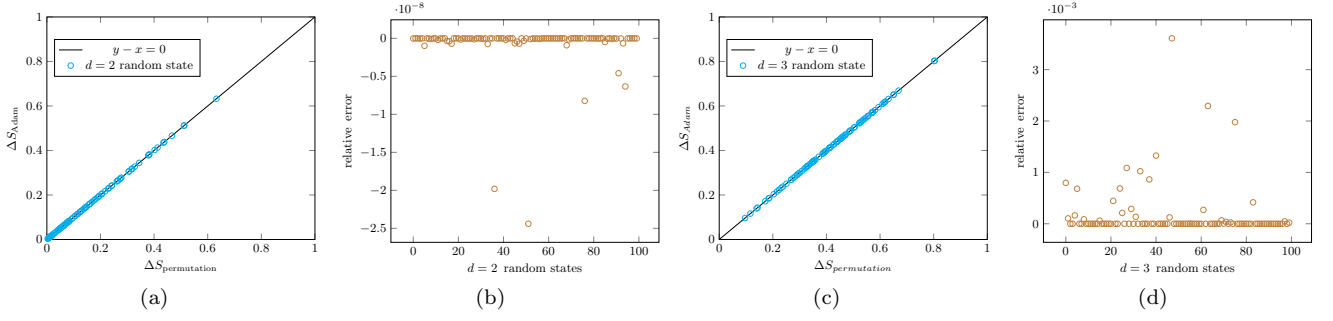
\begin{figure*}[htbp]
\subfigure[]{ 
\begin{tikzpicture}[align=center,xscale=0.52,yscale=0.52]
\begin{axis}[
    xlabel=$\Delta S_{ \text{permutation}}$,
    ylabel=$\Delta S_{ \text{Adam}}$,
    xmin=0, xmax=1,
    ymin=0, ymax=1,
    legend style={at={(0.3,0.9)},
    anchor=north}]
\addplot [color=black, domain=0:1, mark=none] {1*\x};
\addplot [color=cyan,only marks,mark=o, mark size= 2pt] table {2bit.dat};
\legend{$y-x=0$,$d=2$ random state}
\end{axis}
\end{tikzpicture}}
\subfigure[]{
\begin{tikzpicture}[align=center,xscale=0.52,yscale=0.52]
\begin{axis}
[title={},
    ylabel=relative\enspace error,
    xlabel= {$d=2$\enspace random states},
    legend style={at={(0.3,0.9)},
    anchor=north}]
\addplot [brown,mark=o,only marks] table {2bit_error.dat};
\end{axis}
\end{tikzpicture}}
\subfigure[]{
\begin{tikzpicture}[align=center,xscale=0.52,yscale=0.52]
\begin{axis}[
    xlabel=$\Delta S_{ permutation}$,
    ylabel=$\Delta S_{ Adam}$,
    xmin=0, xmax=1,
    ymin=0, ymax=1,
    legend style={at={(0.3,0.9)},
    	anchor=north}]
\addplot [color=black, domain=0:1, mark=none] {1*\x};
\addplot [color=cyan,only marks,mark=o, mark size= 2pt] table {2trit.dat};
\legend{$y-x=0$,$d=3$ random state}
\end{axis}
\end{tikzpicture}}
\subfigure[]{
\begin{tikzpicture}[align=center,xscale=0.52,yscale=0.52]
\begin{axis}
[title={},
    ylabel=relative\enspace error,
    xlabel= {$d=3$\enspace random states},
    legend style={at={(0.3,0.9)},
    anchor=north}]
\addplot [brown ,mark=o,only marks] table {2trit_error.dat};
\end{axis}
\end{tikzpicture}}
\caption{Numerical results of entropy difference based on optimal permutation and Adam optimization in the two-qubit case and two-qutrit case. (a)(c) The maximized entropy difference obtained by permutation $\Delta S_{\text{permutation}}$ via that obtained by Adam optimization $\Delta S_{\text{Adam}}$, where $100$ states are randomly generated in each case. (b)(d) Relative error of $\Delta S_{\text{permutation}}$ with respect to $\Delta S_{\text{Adam}}$.}
\label{2bit_and_3bit}
\end{figure*}

\subsection{$d=3$ case}

In the $d=3$ case, the optimal permutation unitary transformation that maximize the entropy difference $\Delta S$ can not  be given explicitly in a table similar as Table~\ref{tab:2} in the $d=2$ case. We observe that for different $\{p_{mn}\}$, the optimal permutation is different. This arises from a competition between increasing $S_A$ and decreasing $S_B$ in the case of $d=3$, while the optimal permutation maximizes $S_A$ and minimize $S_B$ simultaneously in the $d=2$ case. The numerical results on optimal permutations in the $d=3$ case implies that we should give priority to maximize $S_A$, which is related with the number partition problem. In fact, all optimal permutations we have found are the optimal solutions from the number partitioning problem.

In the Adam algorithm for the $d=3$ case, we take the Gellman matrices as the generators of the unitary transformations, whose derivatives are given in details in Appendix~\ref{app:two}. We numerically checked that the optimal permutations are all local maximal values.

Numerical results on the comparison of optimized entropy differences between the optimal permutation and the Adam algorithm are shown in Fig.~\ref{2bit_and_3bit}~(c)(d). The relative error can arrive at the number of $10^{-3}$ order, which clearly demonstrates that the optimal permutation is in general not a global maximum of the entropy difference $\Delta S$. In the other hand, the entropy difference from the optimal permutation is a relatively excellent approximation of the global maximum of $\Delta S$.

\subsection{$d\ge 4$ cases}
\label{sec:dge-4-cases}

For the systems with $d\geq4$, as discussed in the paragraph before Sec.~\ref{sec:d=2-case}, it is difficult to find optimal permutation directly. Motivated by the experience in the case of $d=3$, we take two steps to find the optimal permutation: first, find one permutation that maximizes $S(\rho_A)$; second, minimize $S(\rho_B)$ while keeping $S(\rho_A)$ invariant. The convenient way to describe the above two steps by visualizing permutations in Table~\ref{tab:1}, the first step is to make the row sums as equal as possible, which is the same aim as the number partition problem. In the second step, to keep $S(\rho_A)$ invariant, we keep the row invariant for every number. To decrease the entropy of $S(\rho_B)$, we only need to arrange the numbers in every row in the decreasing order. In the following, we will develop an approximate algorithm to find one optimized permutation that maximizes $S(\rho_A)$ in the first step.

Before present our algorithm, we first review the greedy number partitioning~(GNP) algorithm for number partitioning~\cite{Korf1995FromAT}, on which our algorithm is constructed.  In the number partition problem, we aim to partition $n$ numbers into $k$ set such that the sums of every set are as equal as possible. The GNP algorithm can be stated as follows. First sort the numbers in the descending order, and place the largest k numbers into k set. Then process the remaining numbers sequentially, put the next number to a set that the sum of the set is currently smallest. We write the GNP in the form of pseudo code in Algorithm~\ref{alg:one}. This method is also called Longest-Processing-Time-First scheduling. It has an approximation ratio that in the worst case, the largest sum in the greedy partition is at most $\frac{4k-1}{3k}$ times the optimal largest sum~\cite{doi:10.1137/0117039,Xiao2017/04,10.1145/800200.806205,Chen1993ANO}. And the minimum sum is at least $\frac{3k-1}{4k-2}$ times the optimal smallest sum~\cite{doi:10.1137/0603019,CSIRIK1992281,WU2005407}.

\begin{algorithm}
\caption{Greedy Number Partitioning (GNP)}\label{alg:one}
\KwIn{$n$ numbers and the number of partitions $k$}
\KwOut{$k$ sets, with an indefinite number of numbers in each set}
sort $n$ numbers in descending order\;
create $k$ empty sets\;
\For{$num$ {\rm in the first} $k$ {\rm largest numbers}}
{append $num$ to one empty set\;}
\For{$num$ {\rm in the remained} $n-k$ {\rm numbers}}
{sum each set, and append $num$ to the set that has the smallest sum\;}
\For{{\rm each set}}
{sort set in descending order\;}
\Return $k$ sets
\end{algorithm}

Note that in the GNP algorithm, the size of every set is not required to be the same. For our problem, however, we need to partition $d^{2}$ numbers, $\{p_{mn}\}$,  into $d$ rows, every row has exactly $d$ numbers. The relative number partion problem can be stated as follows.  To partition $k_Ak_B$ numbers into $k_A$ sets with each set having $k_B$ numbers, such that the sum of numbers in every set as equal as possible. In our algorithm, we recurrently call function GNP, and we name it recurrent greedy number partitioning (RGNP). The procedure of the RGNP is given as follows. In the first partitioning iteration, by applying GNP to $k_A k_B$ numbers, we get $k_0$ sets having $k_B$ numbers, $k_0^+$ sets having more than $k_B$ numbers, $k_0^-$ sets having less than $k_B$ numbers. We cut all the smaller numbers from any set having more than $k_B$ numbers until the set remains $k_B$ numbers. Keep the sets having exact $k_B$ numbers remained. Then mix all cut numbers with sets that have less than $k_B$ numbers, and using GNP to partition them into $k_0^-$ sets. Then we get $k_1$ sets having $k_B$ numbers, $k_1^+$ sets having more than $k_B$ numbers, $k_1^-$ sets having less than $k_B$ numbers. Then repeat the procedure above, until the $j$-th iteration, $ k_j^+  =0 $, the iteration is finished. We write RGNP in the form of pseudo code in Algorithm~\ref{alg:two}. In Appendix~\ref{app:three}, we give an example for the RGNP algorithm.

\begin{algorithm}
\caption{Recurrent Greedy Number Partitioning (RGNP)}\label{alg:two}
\KwIn{$k_A k_B$ eigenvalues of initial state}
\KwOut{$k_A$ sets with each set containing $k_B$ numbers in descending order}
$partition_{done} = {\rm empty~set}$\;
$set_{todo} =k_A k_B$ eigenvalues of initial state\;
$k, ~k_+ = 0$\;
$k_- = k_B$\;
$i = 0$\;
\While{$k_+ \neq 0$}{
$i$ \KwTo $i+1$\;
$partition^i$ = GNP($set_{todo}$, $k_-$)\;
$set_{todo} = {\rm empty~set}$\;
$k, ~k_+, ~k_- = 0$\;
\For{set {\rm in} $partition^i$}{
  \If{{\rm the number of elements in} set $\le$ $k_B$}{append numbers in $set$ into $set_{todo}$\; $k_- \gets k_- + 1$\;}
\If{{\rm the number of elements in} set $\ge$ $k_B$}{remain the first largest $k_B$ numbers, cut other numbers\; append cut numbers into $set_{todo}$, append the set containing remained numbers into $partition_{done}$\; $k_+ \gets k_+ + 1$\;}
\If{{\rm the number of elements in} set $=$ $k_B$}{append $set$ into $partition_{done}$\; $k \gets k + 1$\;}
}
}

\Return $partition_{done}$
\end{algorithm}

The time complexity of greedy number partitioning is $O(d^2\log d^2)$~\cite{Korf1995FromAT}, where $d$ is the dimension of one subsystem. As for Gradient Descending algorithm, there is not a general expression $O(\cdot)$ of time complexity. But we know that we need to optimize $d^2\cdot d^2-1$, which is $O(d^2\cdot d^2)$ independent parameters using GD, which is much slower than greedy number partitioning.
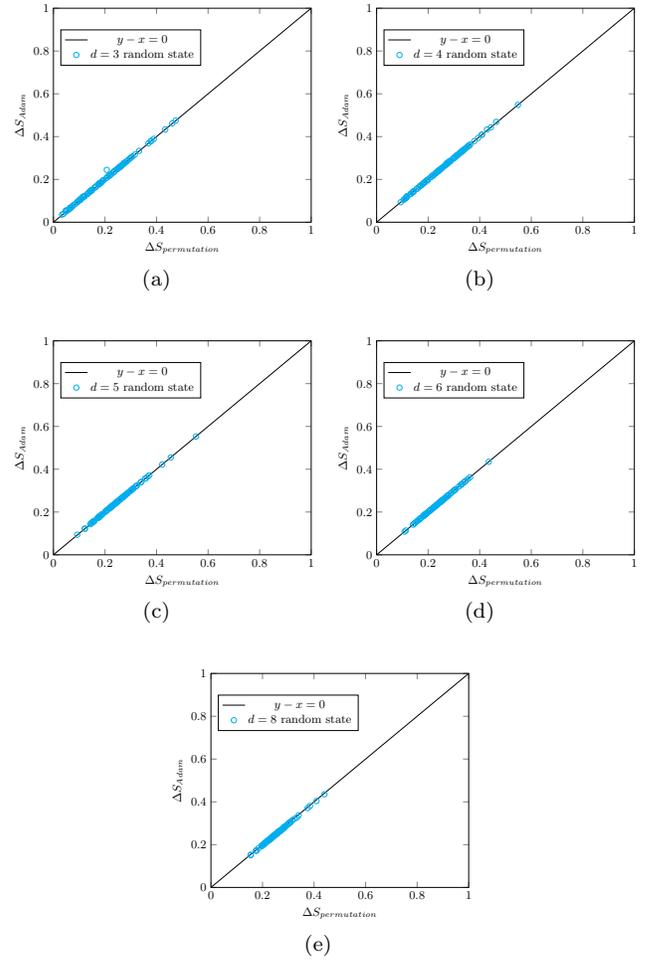
\begin{figure}[htbp]
\subfigbottomskip=10pt
\subfigure[]{
    \begin{tikzpicture}[align=center,xscale=0.5,yscale=0.5]
    \begin{axis}[xlabel=$\Delta S_{ permutation}$, ylabel=$\Delta S_{ Adam}$,xmin=0, xmax=1,ymin=0, ymax=1, legend style={at={(0.3,0.9)},anchor=north}]
    \addplot [color=black, domain=0:1, mark=none] {1*\x};
    \addplot [color=cyan,only marks,mark=o, mark size= 2pt] table {2trit_greedy.dat};
    \legend{$y-x=0$,$d=3$ random state}
    \end{axis}
    \end{tikzpicture}}\label{2trit_greedy}
\subfigure[]{
    \begin{tikzpicture}[align=center,xscale=0.5,yscale=0.5]
    \begin{axis}[xlabel=$\Delta S_{ permutation}$, ylabel=$\Delta S_{ Adam}$,xmin=0, xmax=1,ymin=0, ymax=1,legend style={at={(0.3,0.9)},anchor=north}]
    \addplot [color=black, domain=0:1, mark=none] {1*\x};
    \addplot [color=cyan,only marks,mark=o, mark size= 2pt] table {4bit.dat};
    \legend{$y-x=0$,$d=4$ random state}
    \end{axis}
    \end{tikzpicture}}\label{4bit_greedy}
\subfigure[]{
    \begin{tikzpicture}[align=center,xscale=0.5,yscale=0.5]
    \begin{axis}[xlabel=$\Delta S_{ permutation}$, ylabel=$\Delta S_{ Adam}$,xmin=0, xmax=1,ymin=0, ymax=1,legend style={at={(0.3,0.9)},anchor=north}]
    \addplot [color=black, domain=0:1, mark=none] {1*\x};
    \addplot [color=cyan,only marks,mark=o, mark size= 2pt] table {5level.dat};
    \legend{$y-x=0$,$d=5$ random state}
    \end{axis}
    \end{tikzpicture}}\label{5level_greedy}
\subfigure[]{
    \begin{tikzpicture}[align=center,xscale=0.5,yscale=0.5]
    \begin{axis}[xlabel=$\Delta S_{ permutation}$, ylabel=$\Delta S_{ Adam}$,xmin=0, xmax=1,ymin=0, ymax=1,legend style={at={(0.3,0.9)},anchor=north}]
    \addplot [color=black, domain=0:1, mark=none] {1*\x};
    \addplot [color=cyan,only marks,mark=o, mark size= 2pt] table {6level.dat};
    \legend{$y-x=0$,$d=6$ random state}
    \end{axis}
    \end{tikzpicture}}\label{6level_greedy}
\subfigure[]{
    \begin{tikzpicture}[align=center,xscale=0.5,yscale=0.5]
    \begin{axis}[xlabel=$\Delta S_{ permutation}$, ylabel=$\Delta S_{ Adam}$,xmin=0, xmax=1,ymin=0, ymax=1,legend style={at={(0.3,0.9)},anchor=north}]
    \addplot [color=black, domain=0:1, mark=none] {1*\x};
    \addplot [color=cyan,only marks,mark=o, mark size= 2pt] table {6bit.dat};
    \legend{$y-x=0$,$d=8$ random state}
    \end{axis}
    \end{tikzpicture}
}

\caption{Numerical results of entropy difference based on greedy permutation and Adam optimization for $d=3, 4, 5, 6$ and $8$ cases corresponding to (a)$\sim$(e). The maximized entropy difference obtained by permutation $\Delta S_{\text{permutation}}$ via that obtained by Adam optimization $\Delta S_{\text{Adam}}$, where $100$ states are randomly generated in each case.}
\label{6bit_greedy}
\end{figure}

\begin{table}[htbp]
  \centering
  \begin{tblr}{ccccc}
    \hline
     $d=3$ & $d=4$ & $d=5$ & $d=6$ & $d=8$&\\
    \hline
     $1.1378\%$ &$0.2503\%$ &$0.1017\%$ &$0.0266\%$ &$-1.2422\%$& \\
    \hline
  \end{tblr}
  \caption{Relative error of greedy permutation for $d=3, 4, 5, 6$ and $8$ cases. Relative error is defined by $(\Delta S_{\text{Adam}}-\Delta S_{\text{permutation}})/\Delta S_{\text{Adam}}$. Minus relative error means RGNP algorithm is better than Adam optimization algorithm.}\label{tab:3}
\end{table}

We use this method to maximize the entropy difference for $d=3, 4, 5, 6$ and $8$ cases. We generate $100$ random states for each case, and compare with Adam algorithm. Here we set the GD algorithm convergence condition is $\Delta S<1\times10^{-8}$, which we find this condition is sufficient to converge. The results are shown in Fig.~\ref{6bit_greedy}. And relative error is shown in Table.~\ref{tab:3}. We find that the maximized entropy differences agrees well in all the cases, and the average relative error become smaller with the increasing of $d$. In particular, for a larger system, e.g. the systems with $d=8$, the maximized entropy difference from RGNP is larger than those from Adam algorithm. In other words, when the system has a larger dimension, our RGNP is better than the Adam in most cases, which may be attributed to too many local maximums for a larger system such that the Adam finds only one of local maximums in most cases.  The numerical results show that our RGNP is an excellent approximate algorithm to present a protocol to maximize the entropy difference by an analytical unitary transformation.

\section{Summary and discussion}

We introduce a mutual entropy redistribution protocol via a bipartite unitary transformation in a tripartite pure state. We show that if the dimension Hilbert space of Alice is not less than that of Charlie, then the mutual entropy between Charlie and the other two can be completely redistributed into that between Charlie and Alice. Otherwise, the maximization of the mutual entropy between Alice and Charlie via a bipartite unitary transformation becomes complex, especially when the Hilbert space dimension of Alice (and Bob) is large.

Further more we develop an approximate algorithm for the above complex problem, which is based on greedy number partitioning algorithm, and combined with applications of the two basic properties of Von-Neumann entropy, majorization and concavity. Our numerical experiments show that in small system, this algorithm gives a nearly same results as the optimal result using gradient descending algorithm. In large system, our algorithm gives better results with a faster speed. It is the first time that number partitioning algorithm is used for quantum information processing. In practice, this algorithm helps us get the unitary for redistributing mutual information with a relative small time complexity.

However, we do not understand completely why our approximate RGNP algorithm is so successful for our maximization problem. In the RGNP algorithm we first disentangle the initial state $\rho_{AB}^{0}$, then we perform an optimal permutation related unitary transformation to complete the maximization, where the final optimal state keeps disentangled. Our RGNP algorithm implies that there exists little entanglement between $A$ and $B$ in the optimized state, which is possible due to the monogamy of entanglement~\cite{PhysRevA.61.052306}, i.e., the entanglement between $A$ and $B$ decreases some degrees to maximize the mutual entropy between $A$ and $C$. The theoretical analysis of our RGNP algorithm is left as an open problem to be investigated in future.

\begin{acknowledgments}
This work is supported by National Key Research and Development Program of China (Grant No. 2021YFA0718302 and No. 2021YFA1402104), National Natural Science Foundation of China (Grants No. 12075310), and the Strategic Priority Research Program of Chinese Academy of Sciences (Grant No. XDB28000000).
\end{acknowledgments}

\appendix

\onecolumngrid

\section{Proof of Theorem 2}\label{app:one}

In this part, we aim to prove Theorem 2 in the main text. Specifically, we are going to prove Eqs.~\eqref{eq:77} and \eqref{eq:78} by evaluating second derivatives of the function with respect to parameters $h_{mn}$
\begin{align}
  \Delta S(U) & \equiv S(\rho_A^U) - S(\rho_B^U) \notag \\
  & = S(\Tr_B [W\rho_{AB}^{s_5}W^{\dagger}])- S(\Tr_A [W\rho_{AB}^{s_5}W^{\dagger}]),
\end{align}
where $S$ is von-Neumann entropy, $U=WU^*$ is unitary transformation, $U^*$ is the optimal transformation and unitary transformation $W$ can be parameterized as
\begin{align}
  W &= W_{00}(h_{00}) \prod_{j=1}^3 W_{j0}(h_{j0}) \notag \\
  & \phantom{=} \times \prod_{k=1}^3 W_{0j}(h_{0j}) \prod_{m,n=1}^{3} W_{mn}(h_{mn})
\end{align}
with
\begin{equation}
  W_{m n}(h_{m n}) =\exp(ih_{m n}\sigma_m\otimes \sigma_n).
\end{equation}

\begin{proof}
We can take $\eval{\frac{\partial \Delta S}{\partial h_{31}}}_{h_{31}=0}$ and $\eval{\frac{\partial^2 \Delta S}{\partial h_{31}^2}}_{h_{31}=0}$ for example,
\begin{align}
    W_{31}(h_{31})&=\exp(ih_{31}\sigma_3\otimes \sigma_1) \\ \nonumber
    &=\begin{pmatrix}
        \cos h_{31}&i\sin h_{31}&0&0\\
        i\sin h_{31}&\cos h_{31}&0&0\\
        0&0&\cos h_{31}&-i\sin h_{31}\\
        0&0&-i\sin h_{31}&\cos h_{31}
    \end{pmatrix}.
\end{align}
The density matrix $\rho_{AB}^{W_{31}V^*}$ is
\begin{align}
    \rho_{AB}^{W_{31}V^*}(h_{31})=W_{31}\cdot\rho_{AB}^{V^*}\cdot W_{31}^{\dagger},
\end{align}
and the corresponding reduced density matrix $\rho_{A}^{W_{31}V^*}$ and $\rho_{B}^{W_{31}V^*}$ are 
\begin{align}
    &\rho_{A}^{W_{31}V^*}(h_{31})=\begin{pmatrix}
        p_1+p_4&0\\
        0&p_2+p_3
    \end{pmatrix},\\ \nonumber
     &\rho_{B}^{W_{31}V^*}(h_{31})=\begin{pmatrix}
        \frac{1}{2}+(p_1+p_2-\frac{1}{2})\cos(2h_{31})&-i(-1+2p_1+2p_3)\cos(h_{31})\sin(h_{31})\\
        i(-1+2p_1+2p_3)\cos(h_{31})\sin(h_{31})&\frac{1}{2}-(p_1+p_2-\frac{1}{2})\cos(2h_{31})
    \end{pmatrix}.
\end{align}
Then we evaluate $\Delta S(W_{31}(h_{31})U^*)=S(\rho_A^{W_{31}V^*}(h_{31}))-S(\rho_B^{W_{31}V^*}(h_{31}))$. Finally we can obtain 
\begin{align}
\eval{\frac{\partial \Delta S}{\partial h_{31}}}_{h_{31}=0}&=\lim_{h_{31}\to0} \bigg[\frac{\Delta S(W_{31}(h_{31})U^*)-\Delta S(W_{31}(0)U^*)}{h_{31}}\bigg], \\ \nonumber
\eval{\frac{\partial^2 \Delta S}{\partial h_{31}^2}}_{h_{31}=0}&=\lim_{h_{31}\to0} \bigg[\frac{\frac{\partial \Delta S}{\partial h_{31}}(h_{31})-\frac{\partial \Delta S}{\partial h_{31}}(0)}{h_{31}}\bigg].
\end{align}

We evaluate these expressions above by Mathematica 12.0 and we find that all first-order derivatives are equal to $0$ at $W=0$(these expressions are not shown below because they are too long) while diagonal second-order derivatives are less than or equal to $0$:
\begin{align}
    &\frac{\partial^2 \Delta S}{\partial h_{00}^2}=\frac{\partial^2 \Delta S}{\partial h_{01}^2}=\frac{\partial^2 \Delta S}{\partial h_{02}^2}=\frac{\partial^2 \Delta S}{\partial h_{03}^2}=\frac{\partial^2 \Delta S}{\partial h_{10}^2}=\frac{\partial^2 \Delta S}{\partial h_{20}^2}=\frac{\partial^2 \Delta S}{\partial h_{30}^2}=\frac{\partial^2 \Delta S}{\partial h_{33}^2}=0, \label{useless_parameter}\\  \nonumber
    \\ \nonumber
    &\frac{\partial^2 \Delta S}{\partial h_{11}^2}=\frac{\partial^2 \Delta S}{\partial h_{12}^2}=\frac{\partial^2 \Delta S}{\partial h_{21}^2}=\frac{\partial^2 \Delta S}{\partial h_{22}^2}\\ 
    &=\frac{2}{\ln2}\big[(p_1+p_2-p_3-p_4)\ln(\frac{p_3+p_4}{p_1+p_2})+(p_1+p_4-p_2-p_3)\ln(\frac{p_1+p_4}{p_2+p_3})\big]\leq 0, \label{useful_parameter1} \\ \nonumber
    \\
    &\frac{\partial^2 \Delta S}{\partial h_{13}^2}=\frac{\partial^2 \Delta S}{\partial h_{23}^2}=\frac{8(p_1-p_2)(p_3-p_4)\ln[\frac{1-|-1+2p_2+2p_3|}{1+|-1+2p_2+2p_3|}]}{|-1+2p_2+2p_3|\ln(2)}\leq 0,\label{useful_parameter2} \\ \nonumber 
    \\
    &\frac{\partial^2 \Delta S}{\partial h_{31}^2}=\frac{\partial^2 \Delta S}{\partial h_{32}^2}=\frac{8(p_2-p_3)(p_1-p_4)\ln[\frac{1-|-1+2p_1+2p_2|}{1+|-1+2p_1+2p_2|}]}{|-1+2p_1+2p_2|ln(2)}\leq 0.\label{useful_parameter3} \\ \nonumber 
\end{align}
The sign of Eq.~\eqref{useful_parameter1} is determined by 
\begin{align} \label{scaling_method}
    &(p_1+p_2-p_3-p_4)\ln(\frac{p_3+p_4}{p_1+p_2})+(p_1+p_4-p_2-p_3)\ln(\frac{p_1+p_4}{p_2+p_3}) \\ \nonumber
    \leq&(p_1+p_2-p_3-p_4)\ln(\frac{p_3+p_4}{p_1+p_2})+(p_1+p_2-p_3-p_4)\ln(\frac{p_1+p_2}{p_3+p_4}) \\ \nonumber
    =&0.
\end{align}
Assuming $p_1\neq p_2\neq p_3\neq p_4$, all the diagonal second derivatives in Eqs.~\eqref{useful_parameter1},~\eqref{useful_parameter2} and ~\eqref{useful_parameter3} are less than $0$.

As for those off-diagonal derivatives terms, we can repeat the procedure above by setting 
\begin{align}
    W_{mn,m'n'}&=\exp(ih_{mn}\sigma_m\otimes \sigma_n+ih_{m'n'}\sigma_{m'}\otimes \sigma_{n'}).
\end{align}
By evaluating using Mathematica, we find all off-diagonal terms $\frac{\partial^2 \Delta S}{\partial h_{mn}h_{m'n'}}$ are equal to $0$. So the hessian matrix $H(\Delta S)_{m n; m^{\prime} n^{\prime}}$ is negative-definite and $\Delta S(U^*)$ is at its local maximal point. 

During the computation we find that parameters in Eq.~\eqref{useless_parameter} have no effect on reduced density matrix, i.e. for $m,n$ in Eq.~\eqref{useless_parameter}, $W_{mn}\cdot \rho_A(V)\cdot W_{mn}^{\dagger}=\rho_A(V)$ and $W_{mn}\cdot \rho_B(V)\cdot W_{mn}^{\dagger}=\rho_B(V)$. So we can remove such parameters from hessian matrix and only maintain those "useful" parameters that have non-zero second derivatives.
  
\end{proof}

\section{Derivatives of two-qutrit system}\label{app:two}
In this part, we aim to evaluate the derivatives of entropy difference function of two-qutrit system with respect to parameters of parameterized unitary transformation. In $d=3$ cases we cannot give a theorem like theorem 2 because there is no consistent result for different $d=3$ states. Let $d=3$ state after being applied the disentanglement unitary transformation and permutation unitary transformation be 
\begin{align}
    \rho_{qutrit}=s_1|00\rangle\langle00|+s_2|01\rangle\langle01|+s_3|02\rangle\langle02|+...+s_9|22\rangle\langle22|,
\end{align}
where $s_i$ are eigenvalues after permutation.
There are in total 3 kinds of forms for two-qutrit states $\rho_{qutrit}$ . The first form is, taking $\frac{\partial^2 \Delta S}{\partial h_{12}^2}$ for example:
\begin{align}\label{2trit_derivative3}
\frac{\partial^2 \Delta S}{\partial h_{12}^2}=\frac{2}{\ln2}\bigg[
(s_1+s_2-s_4-s_5)\ln(\frac{s_1+s_2+s_3}{s_4+s_5+s_6}) + (-s_1+s_2-s_4+s_5)\ln(\frac{s_1+s_4+s_7}{s_2+s_5+s_8})\bigg].
\end{align}
The second form is, taking $\frac{\partial^2 \Delta S}{\partial h_{13}^2}$ for example:
\begin{align}
   \frac{\partial^2 \Delta S}{\partial h_{13}^2}= {\rm Part1}\cdot {\rm Part2},
\end{align}
where 
\begin{align}\label{2trit_derivative1}
    {\rm Part1}&=s_3s_4+s_3s_5-4s_4s_5-(s_4+s_5)s_6+s_2(-s_3+4s_4+s_6)+s_1(-4s_2-s_3+4s_5+s_6)\\ \nonumber 
    &=4(s_1-s_4)(s_2-s_5)+(s_6-s_3)(s_4+s_5)+(s_1+s_2)(s_3-s_6),\\ 
    {\rm Part2}&=-\frac{2}{\ln2}\frac{\ln\big(\frac{s_1+s_2+s_3+s_4+s_5+s_6-|s_1+s_2+s_3-s_4-s_5-s_6|}{s_1+s_2+s_3+s_4+s_5+s_6+|s_1+s_2+s_3-s_4-s_5-s_6|}\big)}{|s_1+s_2+s_3-s_4-s_5-s_6|}>0.
\end{align}
The third form is, taking $\frac{\partial^2 \Delta S}{\partial h_{18}^2}$ for example:
\begin{align}
    \frac{\partial^2 \Delta S}{\partial h_{18}^2}={\rm Part1}\cdot {\rm Part2},
\end{align}
where 
\begin{align}\label{2trit_derivative2}
    {\rm Part1}&=(s_1+s_2-s_4-s_5)(s_3-s_6)<0,\\ 
    {\rm Part2}&=-\frac{6}{\ln2} \frac{\ln\big(\frac{s_1+s_2+s_3+s_4+s_5+s_6-|s_1+s_2+s_3-s_4-s_5-s_6|}{s_1+s_2+s_3+s_4+s_5+s_6+|s_1+s_2+s_3-s_4-s_5-s_6|}\big)}{|s_1+s_2+s_3-s_4-s_5-s_6|} >0.
\end{align}
 Note that we cannot determine the sign of Eq.~\eqref{2trit_derivative3} and Eq.~\eqref{2trit_derivative1} since different $\{s_i\}$ will result in different signs. However, by plugging the eigenvalues of $\rho_{qutrit}$ applied by optimal permutations into these derivative expressions, we find that these optimal permutations actually make $S(\rho_A)-S(\rho_B)$ take a local maximum.

\section{Examples of recurrent greedy number partitioning}\label{app:three}

In this appendix, we give an example for the recurrent greedy number partitioning algorithm. 
For a $d=6$ state, the eigenvalues are:

$p$=\{0.184233, 0.172701, 0.167875, 0.130484, 0.007168, 0.006866, 0.005525, 0.00415, 0.002577, 0.002313, 0.101274, 0.009832, 0.008887, 0.008416, 0.007561, 0.006997, 0.006116, 0.004571, 0.003275, 0.000357, 0.000128, 0.043433, 0.011384, 0.011262, 0.010695, 0.010573, 0.010166, 0.010124, 0.009745, 0.008469, 0.008223, 0.007007, 0.006151, 0.005061, 0.003894, 0.002506\}

We first partition ${p}$ into 6 sets:
\begin{align}
s_1=&\{0.184233\},\\   \nonumber
s_2=&\{0.172701\},\\   \nonumber
s_3=&\{0.167875\},\\   \nonumber
s_4=&\{0.130484, 0.007168, 0.006866, 0.005525, 0.00415, 0.002577,|\\ \nonumber
&0.002313\},\\   \nonumber
s_5=&\{0.101274, 0.009832, 0.008887, 0.008416, 0.007561, 0.006997,|\\ \nonumber
&0.006116, 0.004571, 0.003275, 0.000357, 0.000128\},\\   \nonumber
s_6=&\{0.043433, 0.011384, 0.011262, 0.010695, 0.010573, 0.010166,|\\ \nonumber
&0.010124, 0.009745, 0.008469, 0.008223, 0.007007, 0.006151, 0.005061, 0.003894, 0.002506\}.
\end{align}

So we get 3 sets having more than 6 numbers. We denote the number of sets having
$6$ numbers in the $i$-th iteration is $k_i$, the number of sets having more than $6$ numbers is $k_i^+$, the number of sets having less than $6$ numbers is $k_i^-$. We cut numbers after the “$|$” and put these cut numbers with $s_1$, $s_2$ and $s_3$, and get a new set $p^1$ to partition, where $k^-_1=3$:

$p^1$=\{0.184233, 0.172701, 0.167875, 0.002313, 0.006116, 0.004571, 0.003275, 0.000357, 0.000128, 0.010124, 0.009745, 0.008469, 0.008223, 0.007007, 0.006151, 0.005061, 0.003894, 0.002506\}.

Then we partition $p^1$ into 3 sets:
\begin{align}
s_1=&\{0.184233, 0.007007, 0.005061, 0.003894, 0.000128\},\\   \nonumber
s_2=&\{0.172701, 0.009745, 0.008223, 0.006116, 0.003275, 0.000357\},\\   \nonumber
s_3=&\{0.167875, 0.010124, 0.008469, 0.006151, 0.004571, 0.002506, 0.002313\}.
\end{align}

We find that $s_2$ has 6 numbers, $s_1$ has 5 numbers and $s_3$ has 7 numbers, so we move the last number in $s_3$ to $s_1$, and finish the recurrent greedy number partitioning. The final result is:

\begin{align}
s_1=&\{0.184233, 0.007007, 0.005061, 0.003894, 0.000128, 0.002313\},\\   \nonumber
s_2=&\{0.172701, 0.009745, 0.008223, 0.006116, 0.003275, 0.000357\},\\   \nonumber
s_3=&\{0.167875, 0.010124, 0.008469, 0.006151, 0.004571, 0.002506\},\\   \nonumber
s_4=&\{0.130484, 0.007168, 0.006866, 0.005525, 0.00415, 0.002577\},\\ \nonumber
s_5=&\{0.101274, 0.009832, 0.008887, 0.008416, 0.007561, 0.006997\},\\   \nonumber
s_6=&\{0.043433, 0.011384, 0.011262, 0.010695, 0.010573, 0.010166\}.
\end{align}

\section{Relative error of RGNP algorithm}\label{app:four}
Fig.~\ref{fig:app:four} gives relative error of the RGNP algorithm for $d=3,4,5,6$ and $8$ cases.
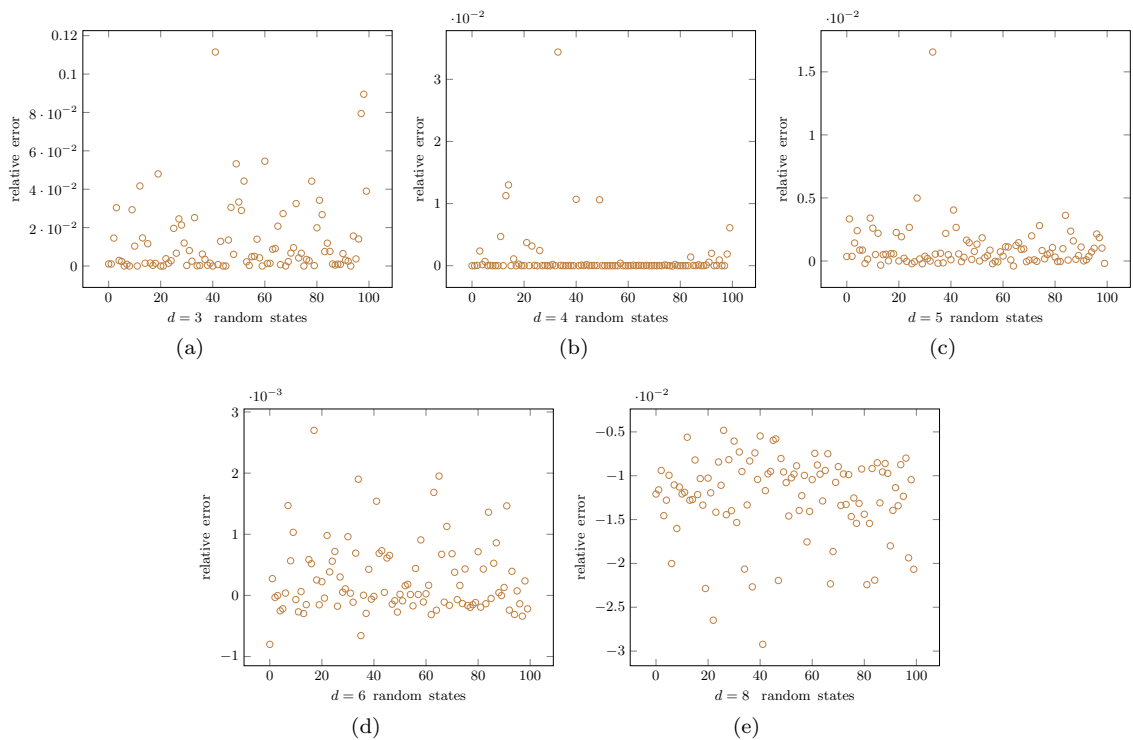
\begin{figure}[htbp]
    \centering
    \subfigure[]{
    \begin{tikzpicture}[align=center,xscale=0.6,yscale=0.6]
    \begin{axis}
    [title={},xlabel={$ d=3 $ \enspace random\enspace states},ylabel=relative\enspace error,legend style={at={(0.3,0.9)},anchor=north,}]
    \addplot [brown ,mark=o,only marks] table {2trit_greedy_error.dat};
    \end{axis}
    \end{tikzpicture}}
\subfigure[]{
    \begin{tikzpicture}[align=center,xscale=0.6,yscale=0.6]
    \begin{axis}
    [title={},xlabel={$ d=4 $\enspace random\enspace states},
    ylabel=relative\enspace error,legend style={at={(0.3,0.9)},anchor=north,}]
    \addplot [brown ,mark=o,only marks] table {4bit_error.dat};
    \end{axis}
    \end{tikzpicture}}
\subfigure[]{
    \begin{tikzpicture}[align=center,xscale=0.6,yscale=0.6]
    \begin{axis}
    [title={},xlabel={$d=5 $\enspace random\enspace states},ylabel=relative\enspace error,legend style={at={(0.3,0.9)},anchor=north}]
    \addplot [brown ,mark=o,only marks] table {5level_error.dat};
    \end{axis}
    \end{tikzpicture}}
\subfigure[]{
    \begin{tikzpicture}[align=center,xscale=0.6,yscale=0.6]
    \begin{axis}
    [title={},xlabel={$ d=6$\enspace random\enspace states},
    ylabel=relative\enspace error,legend style={at={(0.3,0.9)},anchor=north,}]
    \addplot [brown ,mark=o,only marks] table {6level_error.dat};
    \end{axis}
    \end{tikzpicture}}
\subfigure[]{
    \begin{tikzpicture}[align=center,xscale=0.6,yscale=0.6]
    \begin{axis}
    [title={},xlabel={$d=8$ \enspace random\enspace states},ylabel=relative\enspace error,legend style={at={(0.3,0.9)},anchor=north,}]
    \addplot [brown ,mark=o, only marks] table {6bit_error.dat};
    \end{axis}
    \end{tikzpicture}}
    \caption{Relative error of the RGNP algorithm for $d=3,4,5,6$ and $8$ cases. For each case, Table.~\ref{tab:3} writes the average of relative error of 100 states.}
    \label{fig:app:four}
\end{figure}

\twocolumngrid

\bibliographystyle{apsrev4-2} 
\bibliography{bib.bib}

\end{document}